\documentclass[12pt,reqno,a4paper]{amsart}    %A4-format
\usepackage[totalwidth=450pt, totalheight=737pt]{geometry}
\usepackage{colordvi}
\usepackage{enumerate}
\usepackage{amsmath, amssymb}%, theorem}
\usepackage{amsthm}
\usepackage{accents}     % for \accentset, \ring...

\usepackage{graphicx}                      % for graphics
\newcommand{\color}[2][{}]{}         % figure-x.pstex_t file contains this
\usepackage{bbm}         % blackboard (see: /usr/share/texmf/mytex/bbm.dvi)

\usepackage{mathrsfs}    % use ``real'' caligraphic letters for \mathcal
\renewcommand\mathcal\mathscr

\numberwithin{equation}{section}

 \theoremstyle{plain}            % body italics
 \newtheorem{theorem}{Theorem}[section]
 \newtheorem{proposition}[theorem]{Proposition}
 \newtheorem{lemma}[theorem]{Lemma}
 \newtheorem{corollary}[theorem]{Corollary}

 \theoremstyle{definition}       % body roman

 \newtheorem{remark}[theorem]{Remark}
 
\newcommand{\Sec}[1]{Section~\ref{sec:#1}}

\newcommand{\Fig}[1]{Figure~\ref{fig:#1}}

\newcommand{\Footnote}[1]{Footnote~\ref{fn:#1}}

\newcommand{\Thm}[1]{Theorem~\ref{thm:#1}}

\newcommand{\Thmenum}[2]{Theorem~\ref{thm:#1}~(\ref{#2})}
\newcommand{\Thmenums}[3]{Theorem~\ref{thm:#1}~(\ref{#2}) and~(\ref{#3})}

\newcommand{\Lem}[1]{Lemma~\ref{lem:#1}}

\newcommand{\Cor}[1]{Corollary~\ref{cor:#1}}

\newcommand{\Prp}[1]{Proposition~\ref{prp:#1}}
\newcommand{\Prps}[2]{Propositions~\ref{prp:#1} and~\ref{prp:#2}}

\newcommand{\Rem}[1]{Remark~\ref{rem:#1}}

\newcommand{\card}[1]{\lvert#1\rvert}   % from AMS proceedings file
\newcommand{\dd}    {\, \mathrm d}    % not optimal: no \, if at beginning

\DeclareMathOperator{\dom}    {dom}
\DeclareMathOperator{\id}     {id}  % identity map
\DeclareMathOperator{\vol}    {vol}
\newcommand{\specsymb} {\sigma} % symbol for spectrum

\renewcommand{\Re}     {\mathrm {Re}\,}

\newcommand{\spec}[2][{}]   {\specsymb_{\mathrm{#1}}(#2)}

\newcommand{\disspec}[1]{\spec[disc]{#1}}

\newcommand{\Err}{\mathcal O}

\def\Xint#1{\mathchoice
   {\XXint\displaystyle\textstyle{#1}}%
   {\XXint\textstyle\scriptstyle{#1}}%
   {\XXint\scriptstyle\scriptscriptstyle{#1}}%
   {\XXint\scriptscriptstyle\scriptscriptstyle{#1}}%
   \!\int}
\def\XXint#1#2#3{{\setbox0=\hbox{$#1{#2#3}{\int}$}
     \vcenter{\hbox{$#2#3$}}\kern-.5\wd0}}
\def\XXsum#1#2#3{{\setbox0=\hbox{$#1{#2#3}{\int}$}
     \vcenter{\hbox{$#2#3$}}\kern-.60\wd0}}

\newcommand{\dashint}{\Xint-}   % \int with -
\newcommand{\avint}{{\textstyle\dashint}}   % average sum
\newcommand{\R}{\mathbb{R}} % symbol for real numbers
\newcommand{\C}{\mathbb{C}} % symbol for complex numbers
\newcommand{\N}{\mathbb{N}} % symbol for natural numbers
\newcommand{\eps}{\varepsilon} % shortcut
\renewcommand{\phi}{\varphi}   % shortcut
\newcommand{\e}{\mathrm e}  %Euler number
\newcommand{\im}{\mathrm i} % complex unit

\newcommand{\wt}{\widetilde}           % shortcut
\newcommand {\qf}[1]{\mathfrak{#1}}    % font for quadratic forms
\newcommand{\mc}{\mathcal}

\newcommand{\HS}{\mathcal H}           % symbol for Hilbert space

\newcommand{\Sobsymb} {\mathsf H}      % symbol for Sobolev space
\newcommand{\Contsymb} {\mathsf C}     % symbol for cont. space
\newcommand{\Lsymb}    {\mathsf L}     % symbol for int L-spaces
\newcommand{\Sobspace}[1]{\Sobsymb^{#1}}      % symbol for Sobolev space
\newcommand{\Lsqrspace}    {\Lsymb_2}     % symbol for int L-spaces
\newcommand{\Lin}[1]{\mathcal L({#1})}% symbol for bdd linear operators

\newcommand{\Cont}[2][{}]{\Contsymb^{#1}({#2})}
\newcommand{\Lsqr}[2][{}]{\Lsymb_2^{#1}({#2})} % L_2(#1)-spaces
\newcommand{\Sob}[2][1]{\Sobsymb^{#1}({#2})}         % Sobolev space
\newcommand{\abs}[2][{}]{\lvert{#2}\rvert_{{#1}}}    % abs value
\newcommand{\abssqr}[2][{}]{\lvert{#2}\rvert^2_{#1}} % abs squared
\newcommand{\bigabs}[2][{}]{\bigl\lvert{#2}\bigr\rvert_{#1}}     % abs
\newcommand{\bigabssqr}[2][{}]{\bigl\lvert{#2}\bigr\rvert^2_{#1}}% abs squared
\newcommand{\Bigabs}[2][{}]{\Bigl\lvert{#2}\Bigr\rvert_{#1}}     % abs
\newcommand{\norm}[2][{}]{\|{#2}\|_{{#1}}}    % norm
\newcommand{\normsqr}[2][{}]{\|{#2}\|^2_{#1}} % norm squared
\newcommand{\bignorm}[2][{}]{\bigl\|{#2}\bigr\|_{#1}}     % norm
\newcommand{\bignormsqr}[2][{}]{\bigl\|{#2}\bigr\|^2_{#1}}% norm squared
\newcommand{\iprod}[3][{}]{\langle{#2},{#3}\rangle_{#1}}  % inner product

\newcommand{\set}[2]{\{ \, #1 \, | \, #2 \, \} }      % set { #1 | #2 }

\newcommand{\map}[3]{ #1 \colon #2 \longrightarrow #3}    % maps
\newcommand{\bd}  {\partial}                % symbol for boundary of a set
\newcommand{\connsum}{\sqcup} % symbol for connected sum as bin op

\newcommand{\restr}[1]{{\restriction}_{#1}} % symbol for map restriction

\newcommand{\conj}[1]{\overline {{#1}}}       % symbol for complex conjugation
\newcommand{\1}{\mathbf 1}                  % if bbm does'nt work
\renewcommand{\1}{\mathbbm 1}                    % blackboard 1

\newcommand{\quadtext}[1]{\quad\text{#1}\quad}
\newcommand{\qquadtext}[1]{\qquad\text{#1}\qquad}

\newcommand{\de} {\mathord{\mathrm d}} % exterior derivative

\newcommand{\vxeps}{{\eps,v}}
\newcommand{\edeps}{{\eps,e}}

\newcommand{\maxsymb}[1]{\overline{#1}}
\newcommand{\cvol}{c_{\vol}}

\newcommand{\starsymb}[1]{#1^{\mathrm{star}}}
\newcommand{\approxsymb}[1]{#1^{\mathrm{approx}}}

\newcommand{\inl}{{\mathrm{int}}}

\begin{document}
\title[Approximation of vertex couplings by Schr\"odinger
operators]{A general approximation of quantum graph vertex couplings by scaled Schr\"odinger operators on thin branched manifolds}

\author{Pavel Exner}
\address{Department of Theoretical Physics, NPI, Academy of Sciences,
25068 \v{R}e\v{z} near Prague, and Doppler Institute, Czech
Technical University, B\v{r}ehov\'{a}~7, 11519 Prague, Czechia}
\email{exner@ujf.cas.cz}

\author{Olaf Post}      % remove \today in final version
\address{School of Mathematics, Cardiff University, Senghennydd Road,
  Cardiff, CF24 4AG, Wales, UK\newline
  \emph{On leave from:} Department of Mathematical Sciences, Durham
  University, England, UK}%
\email{olaf.post@durham.ac.uk}%
\date{\today \quad \emph{File:}
  \texttt{\jobname.tex}}%, \currenttime h}
%\date{\today}

%------------------------------------------------------------
% Abstract.
%------------------------------------------------------------

\begin{abstract}
  We demonstrate that any self-adjoint coupling in a quantum graph
  vertex can be approximated by a family of magnetic Schr\"odinger
  operators on a tubular network built over the graph. If such a
  manifold has a boundary, Neumann conditions are imposed at it. The
  procedure involves a local change of graph topology in the vicinity
  of the vertex; the approximation scheme constructed on the graph is
  subsequently `lifted' to the manifold. For the corresponding
  operator a norm-resolvent convergence is proved, with the natural
  identification map, as the tube diameters tend to zero.
\end{abstract}

\maketitle

%----------------------------------------------------------------------
%
\section{Introduction}
\label{sec:intro}
%
%----------------------------------------------------------------------

The concept of quantum graph \cite{ekkst:08} serves as a laboratory to
study quantum dynamics in situations when the configuration space has
a complicated topology. At the same time, it is a useful tool in
modelling numerous physical phenomena. To employ its full power, one
should be able to understand the meaning of parameters associated with
vertex coupling in such models, because one can typically associate
many self-adjoint Hamiltonians with the same graph. An old and natural
idea was to select plausible ones with the help of ``fat-graph''
approximations; it was formulated for the first time by
  Ruedenberg and Scherr \cite{ruedenberg-scherr:53} who proposed a
  heuristic Green-formula argument to demonstrate that such a
  shrinking limit would yield the simplest coupling conditions
conventionally labelled as Kirchhoff. After the interest to the
  problem had been renewed about twenty years ago a lot of effort was
  made to establish this limit rigorously, both for one-body
  Schr\"odinger equation~\cite{freidlin-wentzell:93, saito:00,
  rubinstein-schatzman:01, kuchment-zeng:01,exner-post:05, post:06}
and the Ginzburg-Landau dynamics \cite{rubinstein-schatzman:01}.
  The idea of fat graph approximations has also been used in spectral
  geometry by~\cite{colin:86b} in order to show that the first
  non-vanishing eigenvalue of a compact manifold of dimension three or
  higher can have arbitrarily high multiplicity.  Moreover, (rescaled)
  fat graphs and their limits can be used in calculating spectral
  invariants --- see~\cite{grieser:08b} for a survey
  and~\cite{mueller-mueller:06} for an example of a graph with one
  edge and two vertices.  We refer to
the monograph~\cite{post:12} for a detailed discussion of these problems and an
extensive bibliography.

However, already in the seventies investigations of branched
  electromagnetic wave\-guides \cite{mehran:78} indicated that the low
  frequency behavior, closely linked to the shrinking limit, can be
  different for different geometries, and after
  geometrically induced states bound states in Dirichlet tubes were discovered it became clear that the answer depends substantially on what boundary conditions one chooses for operators on the tube-like manifolds constructed over the graph ``skeleton'', and that the
  limiting coupling may \emph{not} be of the Kirchhoff type \cite{exner-seba:89}. More recent investigations investigations have
  shown that this is typically the case for Dirichlet tube networks, cf.~\cite{post:05,
  molchanov-vainberg:07, grieser:08, acf:07, cacciapuoti-exner:07,
  dell-antonio-costa:10}. Here an energy renormalisation is needed and
when one chooses the natural one which consists of subtracting the
lowest transverse eigenvalue which blows up when the tube diameter
$\eps$ tends to zero, a nontrivial limit is achieved provided the fat
graph from which one starts has a threshold resonance.

It is natural to ask whether and how can a shrinking limit produce
\emph{all} the admissible vertex couplings. Far from being a mere
  mathematical conundrum the problem is of practical importance. When
  building network-type objects one is primarily interested in control
  of the transport; the possibility to manipulate the junction
  dynamics is one of the most direct ways to achieve this goal.
  Methods to alternate transport properties of beam splitters by
  changing the junction geometry \cite{chien-chen-luan:06,
    zhan-li-wang:11} or material characteristics
  \cite{takeda-yoshino:03} have been devised for photonic crystals,
  and it is only a matter of time when the fabrication technique
  progress will allow to address analogous technological challenges
  for networks of metallic nanowires or carbon nanotubes, modeled by
  by manifolds with Neumann boundary conditions or without a boundary, respectively, or semiconductor ones for which
  Dirichlet conditions are used.

Furthermore, some applications of junction control can be made
  more specific. For instance, Cheon et
  al.~\cite{cheon-tsutsui-fulop:04} proposed the generalized point
  interaction, i.e.\ the simplest nontrivial graph with two edges and a
  general vertex coupling, as a model gate for quantum computing; in a
  similar way other star graphs can model an arbitrary \emph{qudit}. The
  point is that the geometry of the eigenvalue manifolds of the
  corresponding operators is described by the group $U(n)$ where $n$
  is the number of the connected edges; for a concrete way how the
  corresponding eigenvalue anholonomy can be used in Grover search
  algorithm see, e.g.~\cite{tanaka-miyamoto:07, tanaka-nemoto:10}. To
  implement such proposals with real-world objects, a proper
  understanding of the junction dynamics is again essential.

The aim of the present paper is to provide a complete solution to
  the problem for tubular network manifolds the boundary of which is
  either Neumann or absent. The approximation we are going to
  construct has several ingredients. The first is the use of scaled
  potentials. If the network dynamics is described by the
Laplace-Beltrami operator the limit leads to the Kirchhoff coupling,
hence one has to replace it by a suitable family of Schr\"odinger
operators. One proceeds at that in two steps, first an approximation
is constructed on the graph itself and subsequently it is ``lifted''
to the tubular manifold. In this way we have been able
in~\cite{exner-post:09} to approximate two important coupling types
usually referred to as $\delta$ and $\delta'_{\mathrm s}$.  Referring
to the graph approximation result obtained in \cite{exner-turek:07} we
conjectured existence of such approximation to any vertex coupling
with real coefficients which covers all the couplings invariant with
respect to the time reversal.  We are going to show here that one is
not only able to prove the said conjecture but in fact can do better:
following the ``algebraic'' work done in~\cite{cet:10} we demonstrate
here existence of a ``fat-graph'' approximation for \emph{all}
self-adjoint vertex couplings.  The new idea here is to lift the
indicated, quite subtle approximation from the metric graph level to fat graphs keeping a precise control of the estimates.

Let us recall briefly how the approximation constructed
in~\cite{cet:10} works, a detailed description will be given in
\Sec{graphapprox} below. It has several steps:

%\vspace{-.5em}
\begin{enumerate}%[(i)]  % does not work on my computer ... (Olaf)

\item we change locally the graph topology disconnecting the edges and
  connecting the loose ends by addition finite edges the length of
  which tends to zero. Some of them may be missing, depending on the
  coupling we want to approximate

\item the additional edges will be coupled to the original ones by
  $\delta$ conditions of the strength dependent on the approximation
  parameter. We also add a parameter-dependent $\delta$ interaction to
  the centre of these finite edges

\item in order to accommodate the couplings with non-real coefficients
  we add magnetic fields described by (the tangent components of)
  appropriate vector potentials, also dependent on the approximation
  parameter

\end{enumerate}
%\vspace{-.5em}

The main result of this paper consists of ``lifting'' this
approximation to tubular networks and demonstrating that one can
approximate in this way any self-adjoint vertex coupling. Since the
approximation bears a local character we concentrate our attention on
star graphs having a single vertex; an extension to general graphs
satisfying suitable uniformity conditions can be performed in the same
way as in~\cite{post:12}.

Let us remark in addition that in contrast to the approximation
  of Kirchhoff coupling by Neumann Laplacians on a tubular network the
  limit constructed in this paper is \emph{non-generic},
  cf.~\cite{exner-neidhardt-zagrebnov:01}, and at the same time
  \emph{non-unique}. An example of different limits for the same
  coupling will be mentioned in \Sec{ex.vx}, and one can conjecture
  also existence of significantly different approximation schemes, in
  particular, purely geometric ones~\cite{kuchment-post:pre12}. The
  result of the present paper thus allows us to make several
  conclusions. From the physics point of view it answers affirmatively
  the question whether one can approximate all the vertex couplings
  allowed by the sole requirement of probability current conservation,
  and at the same time, it suggests one possible construction to
  achieve this goal technically. On the mathematics side,
  approximations of different vertex couplings open interesting
  possibilities in connection with the mentioned use of fat graphs in
  calculating spectral invariants.

The paper is organised as follows: In the next section, we outline the
approximation procedure on the graph level.  In \Sec{app.mfd}, we
construct the graph-like manifold model.  Moreover, we introduce the
quadratic forms corresponding to our operators on the graph and the
manifolds and relate them with the ``free'' operators, i.e the
corresponding Laplacians.  In \Sec{conv.op} we briefly recall the
convergence of operators and forms acting in different Hilbert spaces,
apply the abstract conclusions to our situation here, and demonstrate
our main result expressed in \Thm{main}.  In \Sec{ex}, we present some
examples, including the case of a metric graph embedded in $\R^\nu$
when the manifold model is an $\eps$-neighbourhood of the graph.

%----------------------------------------------------------------------
%
\section{Approximation on the graph level}
\label{sec:graphapprox}
%
%----------------------------------------------------------------------

As we have indicated in the introduction the approximation is
constructed in two steps. First we solve the problem on the graph
level, and the obtained approximation is then ``lifted'' to
network-type manifolds. The first part of this programme was realised
in~\cite{cet:10} and we summarise here the results as a necessary
preliminary.

Any self-adjoint coupling in a vertex of degree $n$ can be expressed
through vertex conditions --- one usually speaks about \emph{admissible} conditions --- which involve the boundary values $f(0),
f'(0)\in\C^n$. They are conventionally written in the form
 % -------------- %
\begin{equation}
  \label{vertex-bc}
  Af(0)+Bf'(0)=0\,,
\end{equation}
 % -------------- %
where $A,\,B$ are $n\times n$ matrices such that the $n\times 2n$
matrix $(A|B)$ has maximum rank and $AB^*$ is Hermitian,
cf.~\cite{kostrykin-schrader:99}.  A pair $(A,B)$ describing a given
coupling is naturally not unique and there are various ways how to
remove the non-uniqueness, see e.g.~\cite{harmer:00, kuchment:04}.
The most suitable for our purpose is the one given by the following
claim proved in \cite{cet:10}; it is simple but it requires an
appropriate graph edge numbering.

 %-------------%
\begin{proposition}
  \label{prp:stform}
  For a quantum graph vertex of degree $n$, the following is valid: \\ [.2em]
  (a) If $S\in\C^{m \times m}$ with $m\leq n$ is a Hermitian matrix and
  $T\in\C^{m \times (n-m)}$, then the equation
 %-------------%
  \begin{equation}
    \label{STform}
    \left(\begin{array}{cc}
        I^{(m)} & T \\
        0 & 0
      \end{array}\right)f'(0)=
    \left(\begin{array}{cc}
        S & 0 \\
        -T^* & I^{(n-m)}
      \end{array}\right)f(0)
  \end{equation}
 % ------------- %
  expresses admissible vertex conditions which make the graph
  Laplacian a self-adjoint operator.\\ [.2em]
 % ------------- %
  (b) Conversely, for any self-adjoint vertex coupling there is a
  number $m\leq n$ and a numbering of edges such that the coupling is
  described by the conditions~\eqref{STform} with uniquely given
  matrices $T\in\C^{m \times (n-m)}$ and $S=S^*\in\C^{m \times m}$. If
  the edge numbering is given one can bring the coupling into the form
  \eqref{STform} by a permutation $(1,\dots,n) \mapsto
  (\Pi(1),\dots,\Pi(n))$ of the edge indices with the matrices $S,T$
  uniquely determined by the permutation $\Pi$.
\end{proposition}
 % ------------- %

Now we can describe the approximation of such a general vertex coupling.
For simplicity we consider a star graph of $n$ semi-infinite edges; in
view of the proposition we may suppose that the wave functions are
coupled according to~\eqref{STform} renaming the edges if necessary.
The construction has two main ingredients. First of all, we have to
change locally the graph topology, adding vertices to the graph as
well as new edges which would shrink to zero in the limit. In this way
one is able to get~\eqref{vertex-bc} with \emph{real} matrices $A,B$; to
overcome this restriction we need to introduce also local magnetic
fields, i.e.\ to place suitable vector potentials at the added edges.

The construction is sketched in \Fig{gen-approx}; we disconnect the
edges of the star graph and connect their loose endpoints by line
segments supporting appropriate operators according to the following
rules:
 % -------------- %
\begin{figure}[ht]
  \begin{center}
    \begin{picture}(0,0)%
      \includegraphics[scale=0.5]{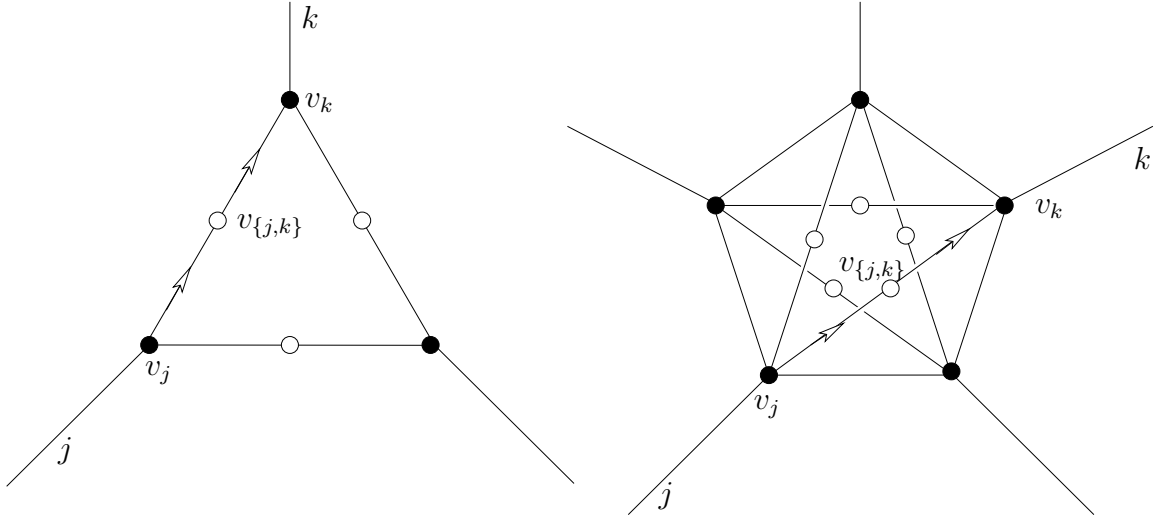}%
    \end{picture}%
    \setlength{\unitlength}{2072sp}%  multiplied with scale factor 0.5
\begin{picture}(13592,6144)(-664,-5473)
  \put(9200,-2550){$v_{\{j,k\}}$}%
  \put(11521,-1861){$v_k$}%
  \put(8191,-4201){$v_j$}%
  \put(7066,-5326){$j$}%
  \put(12691,-1341){$k$}%
  \put(-44,-4786){$j$}%
  \put(2836,344){$k$}%
  \put(2881,-601){$v_k$}%
  \put(2071,-2041){$v_{\{j,k\}}$}%
  \put(991,-3796){$v_j$}%
\end{picture}%
  \end{center}
  \caption{The approximation scheme for a vertex of degree $n=3$ and
    $n=5$.  The inner edges are of length $2d$, some may be missing
    depending on the choice of the matrices $S$ and $T$. The arrows
    symbolise the vector potential.}
  \label{fig:gen-approx}
\end{figure}
% -------------- %

% -------------- %
\begin{enumerate}
  \setlength{\itemsep}{0pt}
\item
  \label{conv.i}
  As a convention, the rows of the matrix $T$ are indexed from 1 to
  $m$, while the columns are indexed from $m+1$ to $n$. For the sake
  of brevity, we use in this section the symbol $\hat n:=
  \{1,\dots,n\}$.

\item The external semi-infinite edges of the approximating graph, each
  parametrised by $s \in\R_+$ are at their endpoints $v_j$ connected
  to the inner edges by $\delta$ coupling with the parameter $w_j(d)$
  for each $j\in\hat{n}$ (see below).

\item
  \label{conv.iii}
  Certain pairs $v_j,v_k$ of external edge endpoints will be connected
  by segments (or \emph{inner edges}, labelled by $\{j,k\}$) of length
  $2d$. This will be the case if one of the following conditions is
  satisfied, taking into account the convention~\eqref{conv.i}:

 % -------------- %
  \begin{enumerate}
    \setlength{\itemsep}{0pt}
  \item $j\in\hat{m}$, $k\geq m+1$, and $T_{jk}\neq0$
    (or $j\geq m+1$, $k\in\hat{m}$, and $T_{kj}\neq0$),
  \item $j,k\in\hat{m}$ and
    $(\exists l\geq m+1)(T_{jl}\neq0\wedge T_{kl}\neq0)$,
  \item $j,k\in\hat{m}$, $S_{jk}\neq0$, and the previous condition is
    not satisfied.
\end{enumerate}
 % -------------- %

\item
  \label{conv.iv}
  We denote the centre of such a connecting segment by
  $v_{\{j,k\}}$ and place there $\delta$ interaction with a parameter
  $w_{\{j,k\}}(d)$. We adopt another convention: the connecting edges
  will be regarded as union of two line segments of the length $d$,
  with the variable running from zero at $w_{\{j,k\}}$ to $d$ at $v_j$
  or $v_k$.

\item Finally, we put a vector potential on each connecting segment.
  What matters is its component tangential to the edge; we suppose it
  is constant along the edge and denote its value between the points
  $v_{\{j,k\}}$ and $v_j$ as $A_{(j,k)}(d)$, and between the points
  $v_{\{j,k\}}$ and $v_k$ as $A_{(k,j)}(d)$; recall that the two
  half-segments have opposite orientation, thus
  $A_{(k,j)}(d)=-A_{(j,k)}(d)$ holds for any pair $\{j,k\}$.
\end{enumerate}
% -------------- %
The choice of the dependence of $w_j(d)$, $w_{\{j,k\}}(d)$, and
$A_{(j,k)}(d)$ on the length parameter $d$ is naturally crucial; we
will specify it below. We denote by $N_j\subset\hat{n}$ the set
containing indices of all the external edges connected to the $j$-th
one by an inner edge, i.e.\
%-------------- %
\begin{eqnarray*} %\label{Nj}
  N_j &\!\!:=\!\!& \{k\in\hat{m}:\, S_{jk}\neq0\}\cup\{k\in\hat{m}:\,
  (\exists l\geq m+1)(T_{jl}\neq0\wedge T_{kl}\neq0)\} \\
  & &\cup\,\{k\geq m+1:\, T_{jk}\neq0\} \hspace{5em} \mbox{for } j\in\hat{m}\\
  N_j &\!\!:=\!\!& \{k\in\hat{m}:\, T_{kj}\neq0\} \hspace{7.7em} \mbox{for } j\geq m+1 \nonumber
\end{eqnarray*}
% -------------- %
The definition of the set $N_j$ has two simple consequences, namely
% -------------- %
\begin{equation*}
  k\in N_j\Leftrightarrow j\in N_k \qquad \mathrm{and} \qquad
  j\geq m+1\Rightarrow N_j\subset\hat{m}\,.
\end{equation*}
% -------------- %
We employ the following symbols for wave function components on the
edges: those on the $j$-th external one is denoted by $f_j$, while the
wave function on the connecting segments is denoted $f_{(j,k)}$ on the
interval between $v_{\{j,k\}}$ and $v_j$ and $f_{(k,j)}$ on the other
half of the segment; the conventions about parametrisation of the
intervals have been specified above.

Next we shall write explicitly the coupling conditions involved in the
above described scheme, first without the vector potentials; for
simplicity we will often refrain from indicating the dependence of the
parameters $w_j(d)$, $w_{\{j,k\}}(d)$ on the distance $d$. The
$\delta$ interaction at the segment connecting the $j$-th and $k$-th
outer edge (present for $j,k \in \hat{n}$ such that $k \in N_j$) is
expressed through the conditions
%------------- %
\begin{equation*}
  f_{(j,k)}(0) = f_{(k,j)}(0) =:f_{\{j,k\}}(0)\,,
  \quad
  f_{(j,k)}'(0+)+f_{(k,j)}'(0+) = w_{\{j,k\}}f_{\{j,k\}}(0)\,,
\end{equation*}
%------------- %
while the $\delta$ coupling at the endpoint of the $j$-th external
edge, $j\in\hat{n}$, means
% ------------- %
\begin{equation*}
  f_j(0)=f_{(j,k)}(d)
  \quadtext{for all  $k\in N_j$,}
  f_j'(0)-\sum_{k \in N_j}f_{(j,k)}'(d-) = w_j f_j(0)\,.
\end{equation*}
% ------------- %
It is not difficult to modify these conditions to include the vector
potentials using a simple gauge
transformation~\cite{cet:10}: the continuity requirement
is preserved, while the coupling parameter changes from $w_j(d)$ to
$w_j(d)+\im \sum_{k\in N_j} A_{(j,k)}(d)$; in other words, the impact
of the added potentials results into the phase shifts $d A_{(j,k)}(d)$
and $d A_{(k,j)}(d)$, respectively, on the appropriate parts of the
connecting segments.

Using the above conditions one can find suitable candidates for
$w_j(d)$, $w_{\{j,k\}}(d)$, and $A_{(j,k)}(d)$ by inserting the
boundary values written as
\begin{gather*}
  f_{(j,k)}(d)
  = \e^{\im  d A_{(j,k)}} \bigr( f_{(j,k)}(0) + d f_{(j,k)}'(0) \bigl)
     + \Err(d^2) \quad\text{and}\\
  f_{(j,k)}'(d) = \e^{\im d A_{(j,k)}} f_{(j,k)}'(0) + \Err(d)
\end{gather*}
for any $j,k \in \hat n$ and fixing the $d$-dependence in such a way
that the limit $d\to 0$ yields~\eqref{STform}. The procedure is
demanding and described in detail in~\cite{cet:10}, we
will mention just its results. As for $A_{(j,k)}(d)$, we have the
relations
\begin{subequations}
  \label{eq:def.w.a}
% -------------- %
\begin{equation}
  \label{a.m.n-m}
  A_{(j,k)}(d)=
  \begin{cases}
    \frac 1 {2d}\arg\,T_{jk} & \text{if $\Re T_{jk}\geq 0$,}\\
    \frac 1 {2d}\bigl(\arg\,T_{jk}-\pi\bigr) & \text{if $\Re
      T_{jk}<0$}
  \end{cases}
\end{equation}
% --------------%
for all $j\in\hat{m}$, $k\in N_j\setminus \hat m$, while for
$j\in\hat{m}$ and $k\in N_j\cap\hat{m}$ we put
% -------------- %
\begin{equation}
  \label{a.m.m-1}
  A_{(j,k)}(d)=
  \begin{cases}
    \frac{1}{2d}\arg \Bigl(d S_{jk}
    +\sum_{l=m+1}^n T_{jl} \conj {T_{kl}} \Bigr) \\
    \frac{1}{2d}\Big[\arg \Bigl(d S_{jk} +\sum_{l=m+1}^n T_{jl}\conj
    {T_{kl}} \Bigr)-\pi\Big]
  \end{cases}
\end{equation}
% --------------%
depending similarly on whether $\Re \bigl(d S_{jk} +\sum_{l=m+1}^n
T_{jl}\conj {T_{kl}} \bigr)$ is nonnegative or not. Concerning
$w_{\{j,k\}}(d)$, we require that
% --------------%
\begin{equation}
  \label{w.m.n-m}
  w_{\{j,k\}}(d)
  =\frac{1}{d}\left(-2+\frac{1}{\langle T_{jk}\rangle}\right)
  \qquad \forall \  j\in\hat{m},\ k\in N_j\setminus\hat{m}\,.
\end{equation}
% --------------%
and
% --------------%
\begin{equation}
  \label{w.m.m}
  \frac{1}{2+d\cdot w_{\{j,k\}}}=-\left\langle d\cdot S_{jk}
    +\sum_{l=m+1}^n T_{jl}\overline{T_{kl}}\right\rangle
  \qquad \forall \ j\in\hat{m},\  k\in
  N_j\cap\hat{m}\,,
\end{equation}
% --------------%
where the bracket symbol on the right hand side is defined as $\langle c\rangle := |c|$ for
$\Re c\ge 0$ and $\langle c\rangle := -|c|$ for $\Re c<0$. Finally, the expressions for
$w_k$ are given by
% --------------%
\begin{equation}
  \label{v.n-m}
  w_k(d)= \frac{1-\card {N_k}+\sum_{h=1}^m\langle T_{hk}\rangle}{d}
  \qquad \forall \  k\geq m+1\,,
\end{equation}
% --------------%
and
% --------------%
\begin{equation}
  \label{v.m}
  w_j(d)
  = S_{jj} -\frac{\card {N_j}}{d}
  - \sum_{\substack{k=1\\k \ne j}}^m
    \Big \langle
      S_{jk}+\frac{1}{d}
      \sum_{l=m+1}^n
      T_{jl}\conj{T_{kl}}
    \Big\rangle
    + \frac{1}{d}
  \sum_{l=m+1}^n(1+\langle T_{jl}\rangle)\langle T_{jl}\rangle
\end{equation}
% --------------%
if $j\in\hat{m}$ and $k\in N_j\cap\hat{m}$.
\end{subequations}

%----------------------------------------------------------------------
\begin{remark}
  \label{rem:order.a.w}
  For our later considerations it is crucial to know precisely the
  dependence of the magnetic and electric potentials
  $A_{(j,k)}=A_{(j,k)}(d)$ and $w_{\{j,k\}}=w_{\{j,k\}}(d)$ on the
  internal length $d$.  We have $A_{(j,k)}(d)=\Err(d^{-1})$ and
  $w_{\{j,k\}}(d)=\Err(d^{-1})$ if $k \in N_j \setminus \hat m$.  If
  $k \in N_j \cap \hat m$, then we have to distinguish two cases. If
  \begin{equation}
    \label{eq:t.term}
    \sum_{l=m+1}^n T_{jl} \conj {T_{kl}} \ne 0,
  \end{equation}
  then we again have $w_{\{j,k\}}(d)=\Err(d^{-1})$.  Otherwise, we
  collect another power of $d^{-1}$ and obtain
  $w_{\{j,k\}}(d)=\Err(d^{-2})$.  We are not aware of any meaning
  of~\eqref{eq:t.term} in terms of the original vertex coupling or
  equivalent characterisations.
\end{remark}
%----------------------------------------------------------------------

The choice of the parameters has been guided by formal considerations
but it opens way to prove the convergence of the corresponding
operators. Let us denote the Laplacian on the star graph $\Gamma(0)$
with the coupling~\eqref{STform} in the vertex as $\starsymb H$, while
$\approxsymb H_d$ will stand for the operators of the described
approximating family; the symbols $\approxsymb R(z)$ and $\approxsymb
R_d(z)$ will denote respectively the resolvents of those operators at
the energy $z$ outside the spectrum. We have to keep in mind that they
act on different spaces: $\starsymb R(z)$ maps $\Lsqr {\Gamma(0)}$
onto $\dom \starsymb H$, while the domain of $\approxsymb R_d(z)$ is
$\Lsqr {\Gamma^{S,T}(d)}$, where $\Gamma^{S,T}(d)=\Gamma(0) \connsum
\Gamma_\inl^{S,T}(d)$ and where $\Gamma_\inl^{S,T}(d)$ is the graph of
connecting (inner) edges of length $2d$ described above. In order to
compare the resolvents, we identify thus $\starsymb R(z)$ with the
orthogonal sum
 % -------------- %
\begin{equation}
  \label{Gdecomp}
  \starsymb R_d(z):=\starsymb R(z)\oplus0
\end{equation}
% -------------- %
adding the zero operator acting on $\Lsqr{\Gamma_\inl^{S,T}(d)}$.
Then both operators act on the same space and one can estimate their
difference; using explicit forms of the corresponding resolvent
kernels one can check in a straightforward but rather tedious way the
relation
% -------------- %
\begin{equation*}
  \norm[\mc B_2] {\starsymb R_d(z)-\approxsymb R_d(z)}
  = \Err(\sqrt{d}) \quadtext{as} d\to 0+
\end{equation*}
for the Hilbert-Schmidt norm, see~\cite{cet:10}. With
the identification~\eqref{Gdecomp} in mind we can then state the
indicated approximation result.

% ------------- %
\begin{theorem}
  \label{thm:ce-ex-tu}
  Let $w_j(d)$, $j \in \hat{n}$, $w_{\{j,k\}}(d)$, $j \in \hat{n}$,
  $k\in\N_j$ and $A^{(j,k)}(d)$ depend on the length $d$ according to
  \eqref{a.m.n-m}--\eqref{v.m}.  Then the family
  $H^\mathrm{approx}_d$ converges to $H^\mathrm{star}$ in the
  norm-resolvent sense as $d\to 0+$.
\end{theorem}
% --------------%

We present some examples of vertex coupling approximations in \Sec{ex.vx}.

%----------------------------------------------------------------------
%
\section{Approximation by Schr\"odinger operators on manifolds}
\label{sec:app.mfd}
%
%----------------------------------------------------------------------

Now we pass to the second step and show how the intermediate quantum
graph constructed in~\Sec{graphapprox} with $\delta$ couplings and
vector potentials can be approximated by scaled magnetic Schr\"odinger
operators on manifolds. For the sake of simplicity, we consider first
an approximation using abstract manifolds without boundary, and
discuss the case of a graph embedded in $\R^\nu$ subsequently
in~\Sec{emb.gr}. To set up the approximation scheme, it is convenient
to work with appropriate quadratic forms instead of the associated
operators.

%----------------------------------------------------------------------
\subsection{The spaces and quadratic form on the graph level}
\label{sec:spaces}
%----------------------------------------------------------------------

We start with the definition of the Hilbert space and quadratic form
on the intermediate graph $\Gamma=\Gamma^{S,T}(d)$, where $d \in
(0,1]$ denotes the approximation parameter of the previous section. It
is convenient to modify slightly the convention~\eqref{conv.iv}
concerning the internal edges $e=\{j,k\}$; from now on we shall
consider each of them as a single edge with the $\delta$ interaction
in the middle (i.e.\ at $v_{\{j,k\}}$) and identify this edge with the
interval $[-d,d]$, oriented in such a way that the parameter increases
from $j$ to $k$ if $j<k$.  Concerning the vector potential, we set
$A_e := A_{(j,i)} = - A_{(i,j)}$. For the sake of brevity, we use the
symbols $A=(A_e)_e$, $w=(w_e,w_v)_{e,v}$ for the collections of
magnetic potentials and $\delta$ interaction strengths, respectively.
We will also often suppress in the sequel the dependence of the
quantities on $d$, $A$, and $w$. With each outer edge $e \in \hat
n=\{1,\dots, n\}$, we associate $I_e := [0,\infty)$, and for each
inner edge $e \in
\begin{pmatrix} \hat n\\ 2 \end{pmatrix}=\set{\{j,k\}} {1 \le j < k
  \le n}$, we set $I_e=I_e(d)=[-d,d]$.  As the Hilbert and Sobolev
spaces on a fixed edge needed in our approximation we set
% -------------- %
\begin{equation*}
  \HS_e := \Lsqr {I_e} \qquadtext{and}
  \HS_e^1 := \Sob{I_e},
\end{equation*}
% -------------- %
where $\Lsqr I$ and $\Sob I$ denote as usual the space of square
integrable functions and of once weakly differentiable and square
integrable functions on the interval $I$, respectively.

For all the quadratic forms defined below, the domains consist of
elements of $\HS_e^1$.  With the described parametrisation of an inner
edge $e=\{j,k\}$ with $i<k$ the corresponding quadratic form is
% -------------- %
\begin{equation*}
  \check{\qf h}_e(f_e)
  := \int_{-d}^d \bigabssqr{f_e'(s) + \im A_e f_e(s)}\dd s
        +  w_e \bigabssqr{f_e(0)}.
\end{equation*}
% -------------- %
This form corresponds to the Laplacian on the edge with the magnetic
potential $A_e$ and the $\delta$ interaction at the point $s=0$.  It
is convenient to introduce also a quadratic form which includes the
effect of the $\delta$ interactions at the edge endpoints, namely
% -------------- %
\begin{equation*}
  \qf h_e(f_e)
  := \check{\qf h}_e(f_e) + \frac{w_j}{\card {N_j}} \cdot \bigabssqr{f_e(-d)}
                  + \frac{w_k}{\card {N_k}} \cdot \bigabssqr{f_e(d)} .
\end{equation*}
% -------------- %
On an outer edge, we simply set
% -------------- %
\begin{equation*}
  \qf h_e(f_e) := \int_0^\infty \bigabssqr{f_e'(s)} \dd s.
\end{equation*}
% -------------- %
The full Hilbert and Sobolev spaces are
% -------------- %
\begin{equation*}
  \HS:=\bigoplus_e \HS_e
  \qquadtext{and}
  \HS^1 := \bigoplus_e \HS_e^1 \cap \Cont {\Gamma},
\end{equation*}
% -------------- %
where the sum runs over all the inner and outer edges.  More
explicitly, the Sobolev space $\HS^1$ consists of all functions in
$\Sob{I_e}$ on each edge, which are continuous on $\Gamma$, i.e.\ which
have a common value
% -------------- %
\begin{equation*}
  f(v) := f_e(v) :=
  \begin{cases}
    f_e(0), & \text{if $e=j$ is an outer edge},\\
    f_e(-d), & \text{if $e=\{j,k\} \sim v=v_j$ is an inner
      edge, $j<k$,}\\
    f_e(d), & \text{if $e=\{j,k\} \sim v=v_k$ is an inner edge,}
  \end{cases}
\end{equation*}
% -------------- %
for all edges $e \sim v$, i.e.\ adjacent with $v$.

The quadratic form on the intermediate graph $\Gamma(d)$ is given by
% -------------- %
\begin{equation*}
  \qf h(f) := \sum_e \qf h_e(f_e)
\end{equation*}
% -------------- %
for $f=(f_e)_e \in \HS^1$; the corresponding operator is the one
described in~\Sec{graphapprox} with $\delta$ interactions of strength
$w_j$ at vertex $v_j$ and of strength $w_e$ in the middle of the inner
edge $e=\{j,k\}$, as well as vector potential $A_{(j,k)}$ supported by
this edge.

For comparison reasons, we also need the \emph{free} quadratic form,
without both the magnetic potentials and the $\delta$ interactions,
which is given by
% -------------- %
\begin{equation*}
  \qf d_e(f_e) := \int_{I_e} \abssqr{f_e(s)}\dd s
  \qquadtext{and}
  \qf d(f):=\sum_e \qf d_e(f_e)
\end{equation*}
% -------------- %
with the same domains as $\qf h_e$ and $\qf h$, respectively.  It is
easy to see that $\qf d$ is a \emph{closed} quadratic form, i.e.\ that
$\dom \qf d=\HS^1$ with the norm given by $\normsqr[\HS^1] f :=
\normsqr f + \qf d(f)$ is complete, and therefore itself a Hilbert
space.  The operator corresponding to $\qf d$ is the \emph{free Laplacian} on $\Gamma(d)$, often also called \emph{Kirchhoff Laplacian} on the graph.

%----------------------------------------------------------------------
\begin{proposition}
  \label{prp:h.rel.bdd}
  \indent
  \begin{enumerate}
  \item
    \label{h.rel.bdd.i}
    The quadratic form $\qf h$ is relatively form-bounded with respect
    to $\qf d$ with relative bound zero. More precisely, for any
    $\eta>0$ there is a constant $C_\eta>0$ depending only on $\eta$,
    $d$, $\maxsymb A := \max_e\ \abs{A_e}$, and $\maxsymb w := 3\,
    \max_{e,v}\{\abs {w_e}, \abs {w_v}\}$ such that
    % -------------- %
    \begin{equation*}
      \bigabs{\qf h(f) - \qf d(f)} \le \eta\, \qf d(f) + C_\eta \normsqr f.
    \end{equation*}
    % -------------- %
    In particular, $\qf h$ is also a closed form.
  \item
    \label{h.rel.bdd.ii}
    We have $\qf d(f) \le 2 \bigl( \qf h(f) + C_{1/2} \normsqr f
    \bigr)$.
  \end{enumerate}
\end{proposition}
%----------------------------------------------------------------------
\begin{proof}
  \eqref{h.rel.bdd.i}~On the interval $[-d,d]$ we have the following
  standard estimate
% -------------- %
  \begin{equation}
    \label{eq:sob.est}
    \bigabssqr{f(s)} \le a \normsqr {f'} + \frac 2 a \normsqr f
  \end{equation}
% -------------- %
  for all $s \in [-d,d]$, $0<a \le d$, and $f \in \Sob{-d,d}$.
  Moreover, for any $\eta>0$ and $a,b \in \R$ we have
  % -------------- %
  \begin{equation}
    \label{eq:quad.sum}
    \frac 1 {1+\eta} \cdot a^2 - \frac 1 \eta \cdot b^2
    \le (a+b)^2
    \le (1+\eta) \cdot a^2 + \Bigl(1 + \frac 1 \eta\Bigr) \cdot b^2.
  \end{equation}
% -------------- %
  In particular, for an inner edge $e=\{j,k\}$ we have
% -------------- %
  \begin{multline*}
    \qf h_e (f) - \qf d_e(f)
    = \normsqr{f' + \im A_e f} - \normsqr{f'}
      + w_e \bigabssqr{f(0)}
      + \frac{w_j}{\card {N_j}} \bigabssqr{f(-d)}
      + \frac{w_k}{\card {N_k}} \bigabssqr{f(d)}\\
    \le \Bigl(\frac \eta 2 + \maxsymb w_e a \Bigr) \normsqr{f'}
      + \Bigl( \Bigl(1 + \frac 2 \eta\Bigr) \abssqr{A_e}
        + \frac {2 \maxsymb w_e} a \Bigr)
      \normsqr f
  \end{multline*}
% -------------- %
  on $[-d,d]$ using~\eqref{eq:sob.est} with $s \in \{-d,0,d\}$ and the
  upper estimate in~\eqref{eq:quad.sum} with $\eta/2$ instead of
  $\eta$, where
  \begin{equation*}
    \maxsymb w_e :=
       \abs{w_e}+
       \frac{\abs{w_j}} {\card {N_j}}+
       \frac{\abs{w_k}} {\card {N_k}}\,.
  \end{equation*}
  Choosing
  % -------------- %
  \begin{equation}
    \label{eq:def.a}
    a:= \min \Bigl\{ \frac \eta {2 \maxsymb w_e}, d \Bigr\}
  \end{equation}
  % -------------- %
  we can estimate the coefficient of $\qf d_e(f)=\normsqr {f'}$ by
  $\eta$.

  For the opposite inequality, we have
  % -------------- %
  \begin{equation*}
    \qf d_e(f) - \qf h_e (f)
    \le \Bigl(1 - \frac 1 {1 + \eta/2} + \maxsymb w_e a \Bigr) \normsqr{f'}
      + \Bigl(\frac {2 \abssqr{A_e}} \eta
             + \frac {2 \maxsymb w_e} a \Bigr)
      \normsqr f
  \end{equation*}
% -------------- %
  using now the lower estimate in~\eqref{eq:quad.sum} with $\eta/2$.
  In particular, with $a$ as in~\eqref{eq:def.a} and with
  $1-(1+\eta/2)^{-1}\le \eta/2$ we can again estimate the coefficient
  of $\qf d_e(f)=\normsqr {f'}$ by $\eta$.  As constant $C_{\eta,e}$
  on each edge, we can therefore choose
% -------------- %
  \begin{equation}
    \label{eq:def.c.eta.e}
    C_{\eta,e}
    := \Bigl( 1 + \frac 2 \eta \Bigr) \abssqr{A_e} +
    \max \Bigl\{ \frac {4 \maxsymb w_e^2} \eta,
                 \frac {2 \maxsymb w_e} d
         \Bigr\}.
  \end{equation}
% -------------- %
  Summing up all contributions for each edge, we can choose $C_\eta :=
  \max_e C_{\eta,e}$, and this constant depends only on $\eta$, $d$,
  $\maxsymb A$ and $\maxsymb w$.

  \eqref{h.rel.bdd.ii}~follows with $\eta=1/2$.  In particular,
% -------------- %
  \begin{equation}
    \label{eq:c.eta.dep}
    C_{1/2}=C_{1/2}(d,A,w)
    = \Err \bigl( \maxsymb A^2 \bigr)
    + \Err \bigl( \maxsymb w^2 \bigr)
    + \Err \Bigl( \frac{\maxsymb w} d \Bigr).
  \end{equation}
% -------------- %
\end{proof}
%----------------------------------------------------------------------

%----------------------------------------------------------------------
\subsection{The spaces and quadratic form on the manifold level}
\label{sec:spaces.mfd}
%----------------------------------------------------------------------

We now define the manifold model as in~\cite{exner-post:09}.  For a
given $\eps \in (0, d]$ we associate a connected $(m+1)$-dimensional
manifold $X_\eps$ to the graph $\Gamma(d)$ as follows: To the edge $e$
and the vertex $v$ we associate the Riemannian manifolds
% ------------- %
\begin{equation}
  \label{eq:mfd.ed}
  X_\edeps := I_e \times \eps Y_e \quadtext{and}
  X_\vxeps := \eps X_v,
\end{equation}
% ------------- %
respectively, where $\eps Y_e$ is a manifold $Y_e$ of dimension
$m>0$ (called \emph{transverse manifold)} equipped with the metric
$h_\edeps:=\eps^2 h_e$.  More precisely, the so-called \emph{edge
  neighbourhood} $X_\edeps$ and the \emph{vertex neighbourhood} $\eps
X_\vxeps$ carry the metrics $g_\edeps=\de^2 s + \eps^2 h_e$ and
$g_\vxeps=\eps^2 g_v$, where $h_e$ and $g_v$ are $\eps$-independent
metrics on $Y_e$ and $X_v$, respectively.  We assume that for each
edge $e$ adjacent to $v$, the vertex neighbourhood $X_\vxeps$ has a
boundary component $\bd_e X_\vxeps=\eps \bd_e X_v$ \emph{isometric} to
the scaled transverse manifold $\eps Y_e$.  Fixing such an isometry
and assuming that $X_\vxeps$ has product structure near each of the
boundary components $\bd_e X_\vxeps$, we identify the boundary
component $\bd_v X_\edeps=\{0\} \times \eps Y_e$ of the edge
neighbourhood $X_\edeps$ with $\bd_e X_\vxeps$.

For simplicity, we assume here that the transversal manifold $Y_e$ has
no boundary and that its volume is normalised, i.e.\ $\vol_m Y_e=1$.

On a Riemannian manifold $X$, we denote by $\Lsqr X$ the Hilbert space
of square integrable functions on $X$ with respect to the natural
measure induced by the Riemannian metric.  Moreover, we denote by
$\Sob X$ the completion of the space of smooth functions with
compact support (not necessarily vanishing on the boundary of $X$)
with respect to the norm given by $\normsqr[\Sob X] u :=
\normsqr[\Lsqr X] u + \normsqr[\Lsqr X]{\de u}$, where $\de u$ denotes
the exterior derivative of $u$ on $X$.

We set
\begin{equation*}
  \HS_\edeps := \Lsqr{I_e,\mc K_\edeps}, \quad
  \mc K_\edeps := \Lsqr{\eps Y_e} \quadtext{and}
  \HS_\vxeps := \Lsqr{X_\vxeps}.
\end{equation*}
We will often identify an $\Lsqrspace$-function $u$ on $X_\edeps$ with
the vector-valued function $I_e \to \mc K_\edeps$, $s \mapsto
u(s):=u(s,\cdot)$.

For each inner edge, we set
\begin{equation*}
  \qf h_\edeps(u_e)
  := \int_{-d}^d
  \bigl(
    \bignormsqr{u_e'(s)+ \im A_e u_e(s)}
    + \qf k_\edeps(u_e(s))
  \bigr) \dd s
  +  \frac{w_e}{2\eps}
  \int_{-\eps}^\eps \bignormsqr{u_e(s)} \dd s,
\end{equation*}
where $u_e'$ denotes the derivative with respect to the longitudinal
variable $s$ and where
\begin{equation*}
  \qf k_\edeps(\phi) := \normsqr[\Lsqr {\eps Y_e}] {\de_{Y_e} \phi}.
\end{equation*}
Here, $\de_{Y_e} \phi$ is the exterior derivative on the manifold
$Y_e$.  For each outer edge we set
\begin{equation*}
  \qf h_\edeps(u_e) :=
  \int_0^\infty \bigl( \normsqr[\mc K_\edeps] {u_e'(s)} + \qf k_\edeps
  (u_e(s)) \bigr) \dd s
\end{equation*}
In both cases, $u_e \in \HS_\edeps^1=\Sob{X_\edeps}$.  On a vertex
neighbourhood, we set
\begin{equation*}
  \qf h_\vxeps(u_e) :=
  \normsqr[\Lsqr {X_\vxeps}] {\de_{X_v} u_v} +
  \frac {w_v}{\eps \vol X_v} \normsqr[\Lsqr {X_\vxeps}] {u_v}.
\end{equation*}

The total Hilbert spaces here are
\begin{equation}
  \label{eq:eps.dec}
  \HS_\eps:=\bigoplus_e \HS_\edeps \oplus \bigoplus_v \HS_\vxeps
  \qquadtext{and}
  \HS^1_\eps := \Sob {X_\eps},
\end{equation}
where the sum runs over all inner and outer edges.
Now, the quadratic form on the manifold $X_\eps$ is given by
\begin{equation*}
  \qf h_\eps(u)
  := \sum_e \qf h_\edeps(u_e) + \sum_v \qf h_\vxeps(u_e)
\end{equation*}
for $u \in \HS^1$ with the obvious notation $u_e:= u \restr{X_\edeps}$
and $u_v := u \restr{X_\vxeps}$.  The corresponding operator is a
magnetic Schr\"odinger operator on $X_\eps$ with (constant) potential
$w_e/(2\eps)$ on $[-\eps,\eps] \times \eps Y_e$ in the middle of an
edge neighbourhood and $w_v/(\eps \vol X_v)$ on each vertex
neighbourhood.  For the use of non-constant potentials we refer to
\cite{exner-post:09}.

For comparison reasons, we also need the \emph{free} quadratic form
(i.e.\ without magnetic and electric potentials), given by $\qf
d_\edeps(u_e)=\normsqr[\Lsqr{X_\edeps}]{\de u_e}$, $\qf
d_\vxeps(u_v)=\normsqr[\Lsqr{X_\vxeps}]{\de u_v}$ and
\begin{equation*}
  \qf d_\eps(u):= \normsqr[\Lsqr {X_\eps}] {\de u}
  = \sum_e \qf d_\edeps(u_e) + \sum_v \qf d_\vxeps(u_v)
\end{equation*}
with the same domains as for $\qf h_\edeps$, $\qf h_\vxeps$ and $\qf
h_\eps$.  Since we define $\HS^1 = \Sob{X_\eps}$ as the completion of
smooth functions with compact support with respect to the norm
$\normsqr[\HS^1_\eps] u := \qf d_\eps(u) + \normsqr u$, the quadratic
form $\qf d_\eps$ is \emph{closed}.  The operator corresponding to
$\qf d$ is the \emph{Laplacian} on $X_\eps$.

%----------------------------------------------------------------------
\begin{proposition}
  \label{prp:h.rel.bdd.mfd}
  \indent
  \begin{enumerate}
  \item
    \label{h.rel.bdd.mfd.i}
    The quadratic form $\qf h_\eps$ is relatively form-bounded with
    respect to $\qf d_\eps$ with relative bound zero.  More precisely,
    for any $\eta>0$ there is a constant $\wt C_\eta \ge C_\eta >0$
    depending only on $\eta$, $d$, $\maxsymb A := \max_e\ \abs{A_e}$,
    $\maxsymb w := 3\,\max_{e,v}\{\abs {w_e}, \abs {w_v}\}$ and $X_v$
    such that
    \begin{equation}
      \label{eq:h.rel.bdd.mfd}
      \bigabs{\qf h_\eps(u) - \qf d_\eps(u)} \le
      \eta\,\qf d_\eps(u) + \wt C_\eta \normsqr u
    \end{equation}
    for all $0 < \eps \le \eps_0$, where $\eps_0:= \eta
    c(v)/\abs{w_v}$ and where $c(v)$ is a constant depending only on
    $X_v$.  In particular, $\qf h_\eps$ is also a closed form.
  \item
    \label{h.rel.bdd.mfd.ii}
    We have $\qf d_\eps(u) \le 2 \bigl( \qf h_\eps(u) + \wt C_{1/2}
    \normsqr u \bigr)$.
  \end{enumerate}
\end{proposition}
%----------------------------------------------------------------------
\begin{proof}
  The proof is very similar to the one of
  \Prp{h.rel.bdd}. For~\eqref{h.rel.bdd.mfd.i}, we have the following
  vector-valued version of~\eqref{eq:sob.est}, namely,
  \begin{equation}
    \label{eq:sob.est.mfd}
    \bignormsqr[\mc K_\edeps]{u_e(s)}
    \le a \normsqr[\HS_\edeps] {u_e'} + \frac 2 a \normsqr[\HS_\edeps] {u_e}
  \end{equation}
  for all $s \in [-d,d]$, $0<a \le d$ and $u \in \Sob{X_\edeps}$.  In
  particular, for an inner edge $e=\{j,k\}$ we have
  \begin{equation*}
    \bigabs{\qf h_\edeps (u_e) - \qf d_\edeps(u_e)}
    \le \eta \normsqr[\HS_\edeps]{u_e'}
      + C_{\eta,e} \normsqr[\HS_\edeps]{u_e}
  \end{equation*}
  with $C_{\eta,e}$ as in \eqref{eq:def.c.eta.e}.

  On a vertex neighbourhood, we have
  \begin{multline*}
    \bigabs{\qf h_\vxeps (u_v) - \qf d_\vxeps(u_v)}
    = \frac {\abs{w_v}}{\eps \vol X_v} \normsqr[\HS_\vxeps] {u_v}\\
    \le \frac {\abs{w_v}}{\eps \vol X_v}
    \Bigl(\eps^2 C(v) \normsqr[\Lsqr{X_\vxeps}]{\de u_v}
       + 4\eps \cvol(v) \sum_{e \sim v}
       \Bigl(
          a \normsqr[\HS_\edeps]{u_e'}
          + \frac 2 a \normsqr[\HS_\edeps]{u_e}
       \Bigr)
    \Bigl)
  \end{multline*}
  for $0 < a \le d$ using~\cite[Lem.~2.9]{exner-post:09}, where
  $\cvol(v) := \vol X_v/\vol_m{\bd X_v}$ and $C(v)$ is another
  constant depending only on $X_v$, see~\cite{exner-post:09} for
  details.  Setting
  \begin{equation*}
    a := \min\{d, \eta \vol X_v/(4 \cvol(v) \abs{w_v} \}
    \quadtext{and}
    \eps_0 := \min_v \frac{\vol X_v}{\abs{w_v} C(v)},
\end{equation*}
and summing up all contributions, we can choose $\wt C_\eta >0$ such
that~\eqref{eq:h.rel.bdd.mfd} holds for all $0 < \eps \le \eps_0$ with
  \begin{equation}
    \label{eq:c.eta.dep.mfd}
    \wt C_\eta=\wt C_\eta(d,A,w)
    = \Err \Bigl( \maxsymb A{}^2 \Bigl(1 + \frac 1 \eta \Bigr) \Bigr)
    + \Err \Bigl( \frac{\maxsymb w{}^2} \eta \Bigr)
    + \Err \Bigl( \frac{\maxsymb w} d \Bigr)
  \end{equation}
  and the error term depend additionally only on $X_v$.  The remaining
  assertion~\eqref{h.rel.bdd.mfd.ii} follows as before.
\end{proof}
%----------------------------------------------------------------------

%----------------------------------------------------------------------
%
\section{Convergence of the operators}
\label{sec:conv.op}
%
%----------------------------------------------------------------------

%----------------------------------------------------------------------
\subsection{Norm convergence of operators and forms acting in
  different Hilbert spaces}
\label{sec:quasi.uni}
%----------------------------------------------------------------------
Let us briefly review the concept of norm convergence of operators
acting in different Hilbert spaces introduced first
in~\cite[App.]{post:06}.  A general spectral theory for quasi-unitary
equivalent operators is developed in a more elaborated version
in~\cite[Ch.~4]{post:12}, see also~\cite{exner-post:09}.

Let $\HS$ and $\HS^1$ be Hilbert spaces such that $\HS^1$ is a dense
subspace of $\HS$ with $\norm[\HS] f \le \norm[\HS^1] f$ and similarly
for $\wt \HS^1 \subset \wt \HS$.  Let $\qf h$ and $\wt{\qf h}$ be
closed, quadratic forms, semi-bounded from below with domain $\HS^1$
and $\wt \HS^1$, respectively.

Let $\delta>0$.  We say that $\qf h$ and $\wt{\qf h}$ are
\emph{$\delta$-quasi-unitarily equivalent}\footnote{\label{fn:que} We
  warn the reader that in~\cite{post:12} the notion
  ``$\delta$-quasi-unitary equivalent'' is defined in a slightly more
  general way (allowing e.g.\ a second identification operator
  $\map{J'}{\wt \HS} \HS$ such that $\norm{J^* - J'} \le \delta$ to
  cover some more general situations).  This should not cause any
  confusion here.}  if there are so-called \emph{identification
  operators}
\begin{equation*}
  \map J \HS {\wt \HS}, \quad
  \map{J^1}{\HS^1} {\wt \HS^1} \quadtext{and}
  \map{J^{\prime 1}}{\wt \HS^1} {\HS^1},
\end{equation*}
such that these operators are \emph{$\delta$-quasi unitary}, i.e.\
\begin{subequations}
  \label{eq:closeness}
  \begin{gather}
    \label{eq:j1}
    \normsqr{Jf - J^1f} \le \delta^2 \normsqr[\HS^1] f,
      \hspace*{4ex}
    \normsqr{J^*u - J^{\prime 1}u}
    \le \delta^2 \normsqr[\wt \HS^1] u,\\
    \label{eq:j.inv}
    \normsqr{J^* J f - f} \le \delta^2 \normsqr[\HS^1] f,
          \hspace*{4ex}
    \normsqr{J J^* u - u} \le \delta^2 \normsqr[\wt \HS^1] u,\\
    \label{eq:j.comm.1}
    \bigabs{ \qf h(J^{\prime 1} u, f) - \wt{\qf h}(u, J^1 f)} \le
    \delta \norm[\wt \HS^1] u \norm[\HS^1] f
  \end{gather}
\end{subequations}
for $f$ and $u$ in the appropriate spaces.  The attribute
\emph{$\delta$-quasi-unitary} refers to the fact that we have a
quantitative generalisation of unitary operators. In particular, if
$\delta=0$, then a $\delta$-quasi-unitary operator is just unitary.

On the operator level, we have the following definition: Denote by $H$
and $\wt H$ the (self-adjoint) operators associated to $\qf h$ and
$\wt{\qf h}$.  We say that $H$ and $\wt H$ are
\emph{$\delta$-quasi-unitarily equivalent} (see again \Footnote{que})
if there is an identification operator $\map J \HS {\wt \HS}$ such
that
\begin{equation}
  \label{eq:q-u.op}
  \bignorm{(\id - J^* J) R^\pm} \le \delta, \quad
  \bignorm{(\id - J J^*) \wt R^\pm} \le \delta
  \quadtext{and}
  \bignorm{J R^\pm - \wt R^\pm J} \le \delta,
\end{equation}
where $\norm \cdot$ denotes the operator norm, and where $R^\pm := (H
\mp \im)^{-1}$ and $\wt R^\pm := (\wt H \mp \im)^{-1}$ denote the
resolvents, respectively.  The resolvent estimates are supposed to
hold for both signs.

We have the following relation between the quasi unitary equivalence
for forms and operators.  For convenience of the reader, we give a
short proof of the first assertion here.  The remaining assertions
follow from the abstract theory developed
  in~\cite[App.~A]{post:06} and~\cite[Ch.~4]{post:12}.
%----------------------------------------------------------------------
\begin{theorem}
  \label{thm:op.q-u-e}
  Let $\delta>0$ and $C \ge 1$.  Assume that $\qf h$ and
  $\wt{\qf h}$ are $\delta$-quasi-unitarily equivalent closed
  quadratic forms such that
  \begin{equation*}
    \normsqr[\HS^1] f \le 2(\qf h(f) + C \normsqr f)
    \quadtext{and}
    \normsqr[\wt \HS^1] u \le 2(\wt{\qf h}(u) + C \normsqr u)
  \end{equation*}
  for all $f \in \HS^1$ and $u \in \wt \HS^1$.  Then the following
  assertions hold:
  \begin{enumerate}
  \item
    \label{op.q-u-e.i}
    The associated operators $H$ and $\wt H$ are $(12 C \delta)$-quasi-unitarily equivalent.
  \item
    \label{op.q-u-e.ii}
    There is a universal constant $c(z)>0$ depending only on $z$ such
    that
    \begin{subequations}
      \begin{gather}
        \label{eq:q-u-e.iia}
        \norm{J(H - z)^{-1} - (\wt H - z)^{-1}J}
        \le c(z) C \delta,\\
        \label{eq:q-u-e.iib}
        \norm{J(H - z)^{-1}J^* - (H_\eps - z)^{-1}} \le c(z) C \delta
      \end{gather}
    \end{subequations}
    for $z \in \C \setminus \R$.  Moreover, we can replace the
    function $\phi(\lambda)=(\lambda-z)^{-1}$ in $\phi(H)=(H-z)^{-1}$
    etc.\ by any measurable, bounded function converging to a constant
    as $\lambda \to \infty$ and being continuous in a neighbourhood of
    $\spec H$.

  \item
    \label{op.q-u-e.iii}
    Assume that $\wt H=H_\eps$ is $\delta_\eps$-unitarily equivalent
    with $H$, where $\delta_\eps \to 0$, then the spectrum of $H_\eps$
    converges to the spectrum of $H$ in the sense that if
      $\lambda_\eps \in \spec {H_\eps}$ and $\lambda_\eps \to
      \lambda$, then $\lambda \in \spec H$, and if $\lambda \in \spec
      H$, then there exists $(\lambda_\eps)_\eps$ such that
      $\lambda_\eps \in \spec {H_\eps}$ and $\lambda_\eps \to 0$.
    The same is true for the essential spectrum.

  \item
    \label{op.q-u-e.iv}
    Assume as before that $\wt H = H_\eps$ is $\delta_\eps$-unitarily
    equivalent with $H$, where $\delta_\eps \to 0$, then for any
    $\lambda \in \disspec H$ there exists a family
    $\{\lambda_\eps\}_\eps$ with $\lambda_\eps \in \disspec {H_\eps}$
    such that $\lambda_\eps \to \lambda$ as $\eps \to 0$.  Moreover,
    the multiplicity is preserved.  If $\lambda$ is a simple
    eigenvalue with normalised eigenfunction $\phi$, then for $\eps$
    small enough there exists a family of simple normalised
    eigenfunctions $\{\phi_\eps\}_\eps$ of $H_\eps$ such that
    % ------------- %
    \begin{equation*}
      \norm[\Lsqr{X_\eps}]{J\phi - \phi_\eps} \to 0
    \end{equation*}
    % ------------- %
    holds as $\eps \to 0$.
  \end{enumerate}
\end{theorem}
%--------------------------------------------------------------------
\begin{proof}
  \eqref{op.q-u-e.i}~From our assumption, we have
  \begin{equation*}
    \normsqr[\HS^1] f
    \le 2 (\qf h(f) + C \normsqr f)
    = 2 \bigabs{\qf h(f) + \normsqr f} + 2 (C-1) \normsqr f.
  \end{equation*}
  Moreover, the first term can be estimated as
  \begin{align*}
    \bigabssqr{\qf h(f) + \normsqr f}
    &\le 2 \bigl(\qf h(f)^2 + \norm f^4\bigr)\\
    &= 2 \bigabs{\qf h(f) - \im \normsqr f}
         \bigabs{\qf h(f) + \im \normsqr f}\\
    &= 2 \bigabs{\iprod {(H \mp \im)f} f}
         \bigabs{\iprod f {(H \mp \im)f}}\\
    &\le 2 \normsqr f \normsqr{(H \mp \im)f}
    \le 2 \norm {(H \mp \im)f}^4
  \end{align*}
  using $\norm{(H \mp \im)^{-1}} \le 1$ at the last step.  In
  particular, we have
  \begin{equation}
    \normsqr[\HS^1] f
    \le (2 \sqrt 2 +2C - 2) \normsqr{(H \mp \im)f}
    \le 4 C \normsqr{(H \mp \im)f}
  \end{equation}
  since $2 \sqrt 2 - 2 \le 2 \le 2C$.  Similarly, we can show the same
  estimate for $u$, and we have
  \begin{equation}
       \label{eq:est.1.2}
    \norm[\HS^1] f
    \le 2 \sqrt C \norm{(H \mp \im)f}
    \quadtext{and}
    \norm[\wt \HS^1] u \le 2 \sqrt C \norm{(\wt H \mp \im)u}.
  \end{equation}
  Therefore, we conclude
  \begin{equation*}
    \norm{f - J^* J f}
    \le \delta \norm[\HS^1] f
    \le 2 \sqrt C\delta \norm{(H \mp \im)f}
  \end{equation*}
  by~\eqref{eq:j.inv}, and in particular, $\norm{(\id - J^* J)R^\pm}
  \le 2 \sqrt C \delta$.  The second norm estimate
  in~\eqref{eq:q-u.op} follows similarly.

  For the last norm estimate of the quasi-unitary equivalence of the
  operators in~\eqref{eq:q-u.op}, set $f := R^\pm g \in \dom H$ and $u
  := \wt R^\mp v \in \dom \wt H$.  Then we have
  \begin{align*}
    \iprod {(JR^\pm- \wt R^\pm J)g} v
    &= \iprod {Jf} v - \iprod g {J^* u}\\
    &= \iprod {(J-J^1)f} v
      + \bigl(\iprod {J^1 f} {(\wt H \pm \im) u}
             -\iprod {(H \mp \im) f} {J^{\prime 1} u}
        \bigr) \\
    &\hspace{.45\columnwidth} {}+ \iprod g {(J^{\prime 1}- J^*) u}\\
    &= \iprod  {(J-J^1)f} v
     +  \bigl( \wt{\qf h}(J^1 f, u)
               -\qf h(f, J^{\prime 1} u) \bigr)
     +\iprod g {(J^{\prime 1}- J^*) u}\\
     &\hspace{.3\columnwidth} {}\mp \im \bigl(
             \iprod {(J^1 - J)f} u
              +\iprod f {(J^* - J^{\prime 1}) u}
           \bigr),
  \end{align*}
  and therefore
  \begin{equation}
    \label{eq:est.1.2b}
    \bigabs{\iprod {(JR^\pm - \wt R^\pm J)g} v}
    \le (2 \sqrt C+ 4 C + 3 \cdot 2 \sqrt C) \delta \norm g \norm v
    \le 12 C \delta \norm g \norm v
  \end{equation}
  using~\eqref{eq:closeness} and~\eqref{eq:est.1.2}.

  Once we have the estimates of the quasi-unitary equivalence
  in~\eqref{eq:q-u.op}, the remaining assertions follow as
  in~\cite[App.~A]{post:06} or~\cite[Ch.~4]{post:12}.
\end{proof}
%----------------------------------------------------------------------

We remark that the convergence of higher-dimensional eigenspaces
is also valid, however, it requires some technicalities which we
skip here.

%----------------------------------------------------------------------
\begin{remark}
  \label{rem:why.c}
  Note that we only obtain the quasi-unitary equivalence of the
  operators with a factor $C$ and not $\sqrt C$.  This is due to the
  fact that from~\eqref{eq:j.comm.1}, we collect two factors $2 \sqrt
  C$ for the estimates $\norm[\HS^1] {R^\mp g} \le 2 \sqrt C
  \norm[\HS] g$ and $\norm[\wt \HS^1] {\wt R^\mp v} \le 2 \sqrt C
  \norm[\wt \HS] v$ in~\eqref{eq:est.1.2b}.
\end{remark}
%----------------------------------------------------------------------

%----------------------------------------------------------------------
\subsection{Quasi-unitary equivalence between the graph and manifold forms}
\label{sec:quasi.uni.model}
%----------------------------------------------------------------------

We now apply the abstract results of the previous section to our
problem where
% ------------- %
\begin{align}
  \label{eq:spaces}
    \HS       &:= \Lsqr {\Gamma^{S,T}(d)}, &
    \HS^1     &:= \Sob  {\Gamma^{S,T}(d)}, &
    \wt \HS   &:= \Lsqr {X_\eps}, &
    \wt \HS^1 &:= \Sob  {X_\eps}.
\end{align}
% -------------%
We start with the definition of the identification operator on an
edge. Let
%----------------------------------------------------------------------
\begin{equation*}
  \map {J_e}{\HS_e=\Lsqr{I_e}}{\HS_\edeps=\Lsqr{X_\edeps}}
  \quadtext{be given by}
  J_e f_e := f_e \otimes \1_\edeps,
\end{equation*}
%----------------------------------------------------------------------
where $\1_\edeps$ is the (constant) eigenfunction of $Y_e$ associated
to the lowest (zero) eigenvalue equal to $\eps^{-m/2}$.  Since we
assumed $\vol Y_e=1$, the eigenfunction is normalised.  Its adjoint
acts as transverse averaging,
%----------------------------------------------------------------------
\begin{equation*}
  (J_e^* u_e)(s)=\iprod[\mc K_\edeps] {u_e(s)} {\1_\edeps}
  = \eps^{m/2} \int_{Y_e} u_e(s, y_e) \dd y_e.
\end{equation*}
%----------------------------------------------------------------------

Before defining the global identification operator, we need the
following result:
%----------------------------------------------------------------------
\begin{lemma}
  \label{lem:av.est}
  For $0<d \le 1$, $0 < \eps \le 1$ and $f, g \in \Sob{[-d,d]}$ we
  have
  \begin{equation}
    \label{eq:sob.est2}
    \Bigabs{\frac 1 {2\eps} \int_{-\eps}^\eps f(s) \conj g(s) \dd s -
      f(0) \conj g(0)}
    \le 2 (\eps/d)^{1/2} \norm[\Sobspace 1] f \norm[\Sobspace 1] g.
  \end{equation}
\end{lemma}
%----------------------------------------------------------------------
\begin{proof}
  Note first that
  \begin{equation}
    \label{eq:est1}
    \abssqr{f(s)} \le \frac 2 d \normsqr[\Sobspace 1] f
  \end{equation}
  for $s \in [-d,d]$ by~\eqref{eq:sob.est} since $d\in (0,1]$ by
  assumption.  From $f(s)-f(0) = \int_0^s f'(t) \dd t$ we conclude
  \begin{equation}
    \label{eq:est2}
    \bigabssqr{f(s)-f(0)} \le \abs s \normsqr{f'}.
  \end{equation}
  Now, the left-hand side of~\eqref{eq:sob.est2} can be estimated by
  \begin{multline*}
    \frac 1 {2\eps}
      \int_{-\eps}^\eps \bigabs{f(s)-f(0)} \abs{g(s)} \dd s
      +
    \frac {\abs{f(0)}} {2\eps}
      \int_{-\eps}^\eps \bigabs{g(s)-g(0)} \dd s\\
    \le \frac 1 {2\eps}
    \Bigl(\int_{-\eps}^\eps \abs s \dd s \normsqr{f'}
    \int_{-\eps}^\eps \abssqr{g(s)} \dd s \Bigr)^{1/2}
      + \frac 1 {2\eps}
    \Bigl(\frac 2 d \normsqr[\Sobspace 1] f
         2 \eps\int_{-\eps}^\eps \abs s \dd s \normsqr{g'}
    \Bigr)^{1/2}\\
    \le \frac 1 2 \norm{f'} \sqrt{\frac 2 d}
        \norm[\Sobspace 1] g \sqrt {2\eps}
      + \frac 1 2 \sqrt{\frac 2 d}
        \norm[\Sobspace 1] f \sqrt {2\eps} \norm{g'}
  \end{multline*}
  using~\eqref{eq:est1}--\eqref{eq:est2} together with Cauchy-Schwarz
  inequality, from where the desired estimate follows.
\end{proof}
%----------------------------------------------------------------------

We can now compare the two contributions of the quadratic forms on an
internal edge, including the potential in the middle of this edge.  We
could consider this inner point as a vertex, too, and use the
arguments for vertex neighbourhoods as in~\cite{exner-post:09}.  Since
this vertex has degree two only, we give a direct (and simpler) proof
here:
%----------------------------------------------------------------------
\begin{lemma}
  \label{lem:qf.ed}
  We have
  \begin{equation*}
    \bigabs{ \qf h_\edeps(J_e f_e, u_e) - \check{\qf h}_e(f_e, J_e^* u_e)}
    \le 2\abs{w_e} (\eps/d)^{1/2}\norm[\Sob \Gamma] f \norm[\Sob{X_\eps}] u
  \end{equation*}
  for all $f \in \HS^1=\Sob \Gamma$, $u \in \wt \HS^1 = \Sob{X_\eps}$,
  $0 < \eps \le 1$ and $0 < d \le 1$.
\end{lemma}
%----------------------------------------------------------------------
\begin{proof}
  We have
  \begin{multline*}
    \qf h_\edeps(J_e f_e, u_e) - \check{\qf h}_e(f_e, J_e^* u_e)\\
    = \int_{-d}^d
    \Bigl(
      \iprod[\mc K_\edeps]
         {(f'_e \otimes \1_\edeps + \im A_e f_e \otimes \1_\edeps)(s)}
         {u_e(s)}
      - (f_e'(s) + \im A_e f_e(s))
      \conj{\iprod[\mc K_\edeps] {u_e(s)}{\1_\edeps}}
    \Bigr) \dd s\\
    + w_e
    \Bigl(
      \int_{-\eps}^\eps
         \iprod[\mc K_\edeps]{f_e(s) \1_\edeps} {u_e(s)}
      \dd s
      - f_e(0) \conj{\iprod[\mc K_\edeps]{u_e(0)}{\1_\edeps}}
    \Bigr).
  \end{multline*}
  Note that in the first integral the term with the derivatives and
  the magnetic potential contributions respectively cancel.  Moreover,
  the expression contains no contribution from the transversal
  (sesquilinear) form $\qf k_\edeps$ since $\qf k_\edeps(\1_\edeps,
  \phi)=0$ for any $\phi \in \Lsqr {\eps Y_e}$.  The remaining
  (electric) potential term can be estimated by \Lem{av.est} with
  $f=f_e$ and $g(s)=\iprod[\mc K_\edeps] {u_e(s)}{\1_\edeps}$).
\end{proof}
%----------------------------------------------------------------------
As the global identification operator we define $\map J \HS {\wt \HS}$ by
% ------------- %
\begin{equation*}
%  \label{eq:def.j}
  J f := \bigoplus_e J_e f_e
           \oplus 0
\end{equation*}
% ------------- %
with respect to the decomposition~\eqref{eq:eps.dec}.  In order to
relate the Sobolev spaces of order one we correct the error made at
the vertex neighbourhood by fixing the function to be constant there.
Namely, we define $\map {J^1} {\HS^1} {\wt \HS^1}$ by
% ------------- %
\begin{equation*}
%  \label{eq:j.1}
  J^1 f := \bigoplus_e J_e f_e
           \oplus \eps^{-m/2} \bigoplus_v f(v) \1_v,
\end{equation*}
 % ------------- %
where $\1_v$ is the constant function on $X_v$ with value $1$.  Since
$f$ is continuous on the graph, $Jf$ is continuous along the vertex and
edge neighbourhood boundary, and therefore maps into  the Sobolev
space $\wt \HS^1=\Sob{X_\eps}$.

For the operator $\map {J^{\prime 1}} {\wt \HS^1} {\HS^1}$, we have to
modify $J^*$ in such a way that the first order spaces are respected, namely we set
% ------------- %
\begin{align*}
%  \label{eq:j.1.}
  (J_e^{\prime 1} u)(s):=
  (J_e^* u_e)(s)
    &+ \chi_-(s)  \eps^{m/2} \bigl(\avint_{v_j} u - (J_e^* u_e)(-d)\bigr)\\
    &+ \chi_+(s)  \eps^{m/2} \bigl(\avint_{v_k} u - (J_e^* u_e)(d)\bigr)
\end{align*}
 % ------------- %
on an inner edge $e=\{j,k\}$, $j<k$, where $\chi_\pm$ are smooth
functions with $\chi_\pm(\pm d)=1$, $\abs{\chi_\pm'} \le 2/d$ and
$\chi_\pm(s)=0$ for $\pm s \le 0$.  Moreover,
\begin{equation*}
  \avint_v u
  := \frac 1 {\vol X_v} \iprod {u_v} {\1_v}
  := \frac 1 {\vol X_v} \int_{X_v} u_v \dd x_v
\end{equation*}
is the average of a function $u$ on the (unscaled) vertex
neighbourhood $X_v$.

On an outer edge $e=j$ we set
% ------------- %
\begin{equation*}
  (J_e^{\prime 1} u)(s):=
  (J_e^* u_e)(s)
    + \chi(s)  \eps^{m/2} \bigl(\avint_{v_j} u - (J_e^* u_e)(0)\bigr)
\end{equation*}
% ------------- %
where $\chi$ is a smooth function with $\chi(0)=1$, $\abs{\chi'} \le
2$ and $\chi(s)=0$ for $s \ge 1$.  Note that $J^{\prime 1}$ differs
from $J^*f$ only by a correction near the vertices.  Since $(J^{\prime
  1} u)_e(v)= \eps^{m/2}\avint_v u$ independently of $e \sim v$, the
function $J^{\prime 1} u$ is indeed continuous, and therefore an
element of $\Sob \Gamma$.

Now we can make a claim which represents the main technical ingredient in the analysis of the
  two quadratic forms:
%----------------------------------------------------------------------
\begin{proposition}
  \label{prp:q-u-e}
  Let $0 < d \le 1$, then the quadratic forms $\qf h_\eps$ and $\qf h$
  are $\delta_\eps$-quasi-unitary equivalent, where $\delta_\eps$
  depends on $\eps$, $d$ and $\maxsymb w := 3\,\max_{e,v} \{\abs{w_e},
  \abs{w_v}\}$ as follows
  \begin{equation*}
    \delta_\eps =
        \Err \Bigl( \Bigl(\frac \eps d \Bigr)^{1/2}
                     (\maxsymb w + 1)\Bigr)
        + \Err \Bigl( \frac{\eps^{1/2}} d \Bigr).
   \end{equation*}
   Moreover, the error depends additionally only on $X_v$ and $Y_e$.
\end{proposition}
%----------------------------------------------------------------------
\begin{proof}
  The argument is similar as in the (simpler) situation
    of~\cite[Prop.~3.2]{exner-post:09} using the identification
    operators just defined.  A new feature here is that we employ additionally \Lem{qf.ed}
    to compare the form contribution on the internal edge
    neighbourhood and its counterpart on the metric graph.  Note also that
    in the present situation we have additionally the magnetic potential
    and slightly different constants than
    in~\cite{exner-post:09}.
\end{proof}
%----------------------------------------------------------------------

In particular, by a clever choice of the $\eps$-dependency of the
parameter $d$, we are able to make the following conclusion:
%----------------------------------------------------------------------
\begin{corollary}
  \label{cor:q-u-e}
  Assume that $w_e$, $w_v$ and $A_e$ are chosen as
  in~\eqref{eq:def.w.a}, then $\maxsymb w=\Err(d^{-2})$ and $\maxsymb
  A=\Err(d^{-1})$.  If in addition, $d=\eps^\alpha$ with $0<\alpha <
  1/5$, then $\qf h_\eps$ and $\qf h$ are $\delta_\eps$-quasi
  unitarily equivalent for all $0 < \eps \le \eps_1$, where
  $\delta_\eps=\Err(\eps^{(1-5\alpha)/2})$, and where $\eps_1>0$ is a
  constant.

  Finally, if $\,0<\alpha < 1/13$, then the associated operators
  $H_\eps$ and $H$ are $\wt \delta_\eps$-quasi unitarily equivalent
  with $\wt \delta_\eps = \Err(\eps^{(1-13\alpha)/2})$.
\end{corollary}
%----------------------------------------------------------------------
\begin{proof}
  The quasi-unitary equivalence of the quadratic forms follows from
  \Prp{q-u-e}, as well as the estimate on $\delta_\eps$.  Moreover,
  $\eps_0=\eps_0(\eps)$ as given in \Prp{h.rel.bdd.mfd} is generally
  of order $\Err(1/\maxsymb w)=\Err(\eps^{2\alpha})$, i.e.\ $\eps_0
  \le c \eps^{2\alpha}$.  In particular, we can choose, $\eps_1 =
  c^{1/(1-2\alpha)}$.

  For the last assertion, note that the constants $C_{1/2}$ and $\wt
  C_{1/2}$ of \Prps{h.rel.bdd}{h.rel.bdd.mfd} fulfil $C_{1/2} =
  \Err(\eps^{-4\alpha})$ and $\wt C_{1/2}=\Err(\eps^{-4\alpha})$,
  cf.~\eqref{eq:c.eta.dep} and~\eqref{eq:c.eta.dep.mfd}, since the
  term $\maxsymb w^2=\Err(\eps^{-4})$ is dominant.  The result
  now follows from \Thmenum{op.q-u-e}{op.q-u-e.i} with
  $C:=\max\{C_{1/2},\wt C_{1/2}\}$, and therefore we have $\wt
  \delta_\eps = 12 C \delta_\eps = \Err(\eps^{-4\alpha + (1-5
    \alpha)/2})=\Err(\eps^{(1-13\alpha)/2})$.
\end{proof}
%----------------------------------------------------------------------

Finally we are in position to put the two convergence steps together;
  and to state and prove the main result of this article:
%----------------------------------------------------------------------
\begin{theorem}
  \label{thm:main}
  Assume that $\Gamma(0)$ is a star graph with vertex condition
  parametrised by matrices $S$ and $T$ as in \Sec{graphapprox} and let
  $0 < \alpha < 1/13$.  Then there is a Schr\"odinger operator $H_\eps$
  on an approximating manifold $X_\eps$ as constructed in
  \Sec{spaces.mfd} such that
  \begin{equation*}
    \norm{J \starsymb R_d(z) J^* - R_\eps(z)}
    =\Err(\eps^{\min\{1-13\alpha,\alpha\}/2})
  \end{equation*}
  for $z \in \C \setminus \R$, where $R_\eps(z)=(H_\eps-z)^{-1}$.
\end{theorem}
%----------------------------------------------------------------------
\begin{proof}
  The result is an immediate consequence of \Cor{q-u-e},
  \Thmenum{op.q-u-e}{op.q-u-e.ii} and \Thm{ce-ex-tu}.
\end{proof}
%----------------------------------------------------------------------

%----------------------------------------------------------------------
\begin{remark}
  \label{rem:err.order}
  The error term in the theorem depends only on $z$ and the building
  block manifolds $X_v$ at the vertices and the transversal manifolds
  $Y_e$ on the edges.  If $\alpha=1/14$, we obtain the error estimate
  $\Err(\eps^{1/28})$ which is the maximal value the function $\alpha
  \mapsto \min\{1-13\alpha, \alpha\}/2$ can achieve.

  The error estimate we obtain here is of the same type that we
  obtained in~\cite[Sec.~4]{exner-post:09} when we approximated the
  $\delta'\mathrm{s}$ interaction despite the fact that the present
  approximation of this particular coupling is different, cf.\
  \Sec{ex.vx} below.

  If the condition~\eqref{eq:t.term} mentioned in \Rem{order.a.w} is
  fulfilled for all $j,k$ we obtain a slightly better estimate.  In
  this case, we have $\maxsymb w=\Err(d^{-1})$ instead of
  $\Err(d^{-2})$, and $\qf h_\eps$ and $\qf h$ are
  $\delta_\eps$-quasi, where $\delta_\eps=\Err(\eps^{(1-3\alpha)/2})$.
  Moreover, the associated operators $H_\eps$ and $H$ are $\wt
  \delta_\eps$-quasi unitarily equivalent with $\wt \delta_\eps =
  \Err(\eps^{(1-7\alpha)/2})$. However, both assumptions made about
  $\alpha$, namely $0 < \alpha < 1/13$ and $0 < \alpha < 1/7$, are for
  sure not optimal.
\end{remark}
%----------------------------------------------------------------------

There is an obvious extension to the above convergence result for
quantum graphs $\Gamma_0$ with more than one vertex.  For quantum
graphs with finitely many vertices, the convergence result holds
without changes, and for infinitely many vertices, some uniformity
conditions are needed.  Such questions are discussed in detail
in~\cite{post:06} and~\cite{post:12}.

%----------------------------------------------------------------------
\begin{remark}
  \label{rem:q-u-e}
  One may ask whether one can reformulate the ``quasi-unitary
  equivalence'' for the present situation using
  \begin{equation*}
    \map{\wt J}{\Lsqr {\Gamma(0)}}
    {\Lsqr{\Gamma^{S,T}(d)}
      =\Lsqr {\Gamma(0)} \oplus \Lsqr{\Gamma^{S,T}_\inl(d)}},
    \qquad
    \wt J f = f \oplus 0,
  \end{equation*}
  in which case $\starsymb R_d(z)=\wt J \starsymb R(z) \wt J^*$
  by~\eqref{Gdecomp} and the resolvent convergence of \Thm{ce-ex-tu}
  can be stated as
  \begin{equation}
    \label{eq:q-u-e.2}
    \norm[\Lin{\Lsqr{\Gamma^{S,T}(d)}}]
    {\wt J \starsymb R(z)\wt J^* - \approxsymb R_d(z)}
    = \Err(d^{1/2})
  \end{equation}
  for $d \to 0$. In fact, we are interested primarily in spectral
  consequences of such a reformulation which can be demonstrated in a
  more direct way.  To this end, note that eq.~\eqref{eq:q-u-e.2} is
  just~\eqref{eq:q-u-e.iib} of \Thm{op.q-u-e} \emph{without} the
  constant $C$.  Moreover, from~\cite[Thm.~4.2.9--10]{post:12} one can
  conclude that~\eqref{eq:q-u-e.iia} is valid for more general $\phi$
  than $\phi(\lambda)=(\lambda-z)^{-1}$, see
  \Thmenum{op.q-u-e}{op.q-u-e.ii}.  Using arguments analogous to those
  in \cite[Sec.~4.2--4.3]{post:12}, we can deduce
  from~\eqref{eq:q-u-e.2} that~\eqref{eq:q-u-e.iib} also holds for
  such $\phi$.  Consequently, the spectral convergence stated in
  \Thmenums{op.q-u-e}{op.q-u-e.iii}{op.q-u-e.iv} also holds in this
  situation.
\end{remark}
%----------------------------------------------------------------------

%----------------------------------------------------------------------
%
\section{Examples}
\label{sec:ex}
%
%----------------------------------------------------------------------

%----------------------------------------------------------------------
\subsection{Embedded graphs and graph neighbourhoods}
\label{sec:emb.gr}
%----------------------------------------------------------------------
Consider the situation when the graph is embedded in $\R^\nu$, $\nu
\ge 2$.  This may be a restriction to the vertex coupling if $\nu=2$
and the vertex degree exceeds three; recall that the edges of the
internal graph defined in~\eqref{conv.iii} of \Sec{intro} are supposed
to be non-intersecting. For $\nu \ge 3$ this difficulty can be avoided
in the edges are properly curved. At the same time, irrespective of
$d$ the lengths of the edge parts of the manifold change as $\eps\to
0$ by an amount given by the size of the vertex neighbourhoods. Let us
point out briefly that for such embedded ``fat graphs'' curved and
shortened edges lead to a small error in the approximation only.

Consider first the length change. In our case, the difference of the
original edge length and the one of an embedded edge is of order
$d-\eps=d(1-\eps/d)=d(1-\eps^{1-\alpha})$.  We have shown in
\cite[Lem.~2.7]{exner-post:09} that this leads to an additional error
of order $\Err(\eps^{1-\alpha})$; expressed again in terms of
quasi-unitary operators.

Furthermore, if we allow \emph{curved} edges in the case of a graph
embedded in $\R^\nu$, we still arrive at the same limit operator.  The
error is of order $\Err(\eps^{1-\alpha})$ (see~\cite[Sec.~6.7 and
Prop.~4.5.6]{post:12} for details; the factor $\eps$ comes from the
shrinking rate, the factor $\eps^{-\alpha}$ from the curvature term of
the embedded curve in dimension $\nu=2$; the length shrinks by
$d=\eps^\alpha$ so its curvature is of order $\eps^{-\alpha}$).
Similar arguments apply for $\nu \ge 3$.  In particular, combining the
effect of shortening of edges and curved edges, and using the
transitivity of quasi-unitary equivalence
(\cite[Prop.~4.2.8]{post:12}) we arrive at an error estimate which is
not worse than the one in \Thm{main}.

%----------------------------------------------------------------------
\subsection{Special vertex couplings and approximation by
  Schr\"odinger operators}
\label{sec:ex.vx}
%----------------------------------------------------------------------

While the approximation described in \Thm{main} cover any self-adjoint
coupling, for some of them we have better alternatives. This concerns,
in particular, the $\delta$ coupling where a simple scaled potential
does a better job as explained in \cite{exner-post:09}. On the other
hand, for couplings with functions discontinuous at the vertex we do
not have many alternatives.

It is illustrative to compare the approximation of the
$\delta'_{\mathrm s}$ coupling obtained from the graph-level
approximations described in \Sec{graphapprox} with the one
from~\cite{cheon-exner:04} used in~\cite{exner-post:09} for the
approximation by Schr\"odinger operators.  Recall that a
$\delta'_{\mathrm s}$ coupling of strength $\beta$ in a vertex of
degree $n$ edges characterised by the condition
\begin{equation*}
  \frac 1 \beta J f(0) - f'(0)=0,
\end{equation*}
where $J$ is the $n \times n$ matrix with all entries one. In other
words, the respective $ST$-parametrisation from \Prp{stform} is given
by $m=n$, $S= \beta^{-1} J$ and $T=0$, and the strengths of the
$\delta$ potentials required to approximate $\delta'_{\mathrm s}$
according to \Thm{ce-ex-tu} are
\begin{equation*}
  w_{\{jk\}} = - \frac \beta {d^2} - \frac 2 d
  \quadtext{and}
  w_j =  \frac {2-n} \beta - \frac {n-1} d.
\end{equation*}
In particular, all inner edges are present.  If $n=3$, for instance,
we employ a small triangle graph of length scale $d=\eps^\alpha$
attaching the ``external'' edges to its vertices (as sketched in
\Fig{gen-approx}).  The corresponding Schr\"odinger operator has a
potential of order $-\eps^{-\alpha-1}$ near $v_j$ and of order $-\beta
\eps^{-2\alpha-1}$ at the midpoint of each edge $\{jk\}$; for
simplicity the potentials can be chosen piecewise constant.

The approximation used in~\cite{exner-post:09} is different. Here we
keep the original star graph, but introduce additional
$\delta$-couplings on each edge at distance $d=\eps^\alpha$ of the
central vertex.  The strength of the coupling at the central vertex is
$-\beta/d^2$, hence the Schr\"odinger potential there is of order
$-\beta \eps^{-2\alpha-1}$.  The strength of the coupling at the
additional vertices is $-1/d$, hence the Schr\"odinger potential is of
order $-\beta \eps^{-\alpha-1}$.  One sees that the approximation
graph topology is different but the $\delta$ strengths in the two
cases differ only in lower order terms\footnote{Such differences are
  not unusual, recall the approximations of $\delta'$ on the line in
  \cite{cheon-shigehara:98, exner-neidhardt-zagrebnov:01}; they do not
  matter as long as both choices lead to cancellation of the singular
  terms in the resolvent difference.} with respect to the length scale
$d=\eps^\alpha$.

Let us finally remark that $\delta'_\mathrm{s}$ is not the only
example of interest; our method makes it possible to approximate other
couplings of potential importance such as the scale-invariant ones
analysed recently in~\cite{cet:11}.

%----------------------------------------------------------------------
\subsection*{Acknowledgement}
%----------------------------------------------------------------------

O.P.\ enjoyed the hospitality in the Doppler Institute where a part of
the work was done. The research was supported by the Czech Science
Foundation and Ministry of Education, Youth and Sports within the
projects P203/11/0701 and LC06002.

%----------------------------------------------------------------------
%\bibliographystyle{my-amsalpha}
%\bibliography{/home/post/Aktuell/BibTeX/literatur}
%----------------------------------------------------------------------

\newcommand{\etalchar}[1]{$^{#1}$}
\providecommand{\bysame}{\leavevmode\hbox to3em{\hrulefill}\thinspace}

\end{document}